\documentclass[a4paper,12pt]{article}
\usepackage{color}
\usepackage[mathscr]{eucal}
\usepackage{amsmath,amsfonts,amssymb,amsthm,mathtools}
\usepackage{times}
\usepackage{hyperref}

\advance\textheight\topskip
\numberwithin{equation}{section} \makeatletter
\@addtoreset{equation}{section}

 \newtheorem{prop}{Proposition}[section]
 \newtheorem{lemma}[prop]{Lemma}

\renewcommand{\tilde}{\widetilde}
\renewcommand{\hat}{\widehat}

\newcommand{\bref}[1]{\textbf{\ref{#1}}}

\newcommand{\gh}[1]{\mathrm{gh}(#1)}

\newcommand{\dd}{\partial}
\renewcommand{\d}{\partial}
\renewcommand{\dh}{\mathrm{d_h}}
\newcommand{\dv}{\mathrm{d_v}}

\newcommand{\binner}[2]{%
  {\langle}\kern-4.15pt{\langle}#1{,}\,#2{\rangle}\kern-4.15pt{\rangle}}
\newcommand{\commut}[2]{[#1{,}\,#2]}

\newcommand{\ab}[2]{\big(#1,#2\big)}

\newcommand{\half}{\mathchoice{%
    \ffrac{1}{2}}{\frac{1}{2}}{\frac{1}{2}}{\frac{1}{2}}}

\newcommand{\ffrac}[2]{\raisebox{.5pt}%
  {\footnotesize$\displaystyle\frac{#1}{#2}$}\kern1pt}

\newcommand{\dl}[1]{\mathchoice{\ffrac{\dd}{\dd #1}}{\frac{\dd}{\dd
      #1}}{\ffrac{\dd}{\dd #1}}{\ffrac{\dd}{\dd #1}}}

\newcommand{\st}[2]{{\overset{#1}{#2}}}

\newcommand{\vdl}[1]{\ffrac{{\delta}}{\delta #1}}

\newcommand{\ii}{\mathfrak{i}}

\newcommand{\Liealg}{\mathfrak} 
\newcommand{\algg}{\Liealg{g}}

\newcommand{\fR}{\mathbb{R}}

 \def\cE{\mathcal{E}}
 \def\cH{\mathcal{H}}
 
\newcommand{\Vol}{\mathcal{V}}

\def\BG-Poincare{Barnich:2009jy}
\def\Fedosov-book{Fedosov:1996fu}

\usepackage{cite}
\makeatletter
\g@addto@macro\bfseries{\boldmath}
\makeatother

\usepackage{authblk}
\title{Presymplectic AKSZ formulation of Einstein gravity\vspace{0.5em}}
\author[a,b]{Maxim Grigoriev\vspace{0.75em}}
\author[c]{Alexei Kotov\vspace{0.75em}}
\affil[a]{Lebedev Institute of Physics,
  Leninsky ave. 53, 119991 Moscow, Russia \vspace{1em}}

\affil[b]{Institute for Theoretical and Mathematical Physics,
  Lomonosov Moscow State University, 119991 Moscow, Russia}
  
  \affil[c]{Faculty of Science, University of Hradec Kralove, Rokitanskeho 62, Hradec Kralove 50003, Czech Republic \vspace{1em}}
\date{}                     
\begin{document}

\maketitle

\begin{abstract}
Any local gauge theory can be represented as an AKSZ sigma model (upon parameterization if necessary). However, for non-topological models in dimension higher than 1 the target space is necessarily infinite-dimensional. The interesting alternative known for some time is to allow for degenerate presymplectic structure in the target space. This leads to a very concise AKSZ-like representation for frame-like Lagrangians of gauge systems. In this work we concentrate on Einstein gravity and show that  not only the Lagrangian but also the full-scale Batalin--Vilkovisky formulation is naturally encoded in the presymplectic AKSZ formulation, giving an elegant supergeometrical construction of BV for Cartan-Weyl action. The same applies to the  main structures of the respective Hamiltonian BFV formulation.
\end{abstract}

\tableofcontents

\section{Introduction}

Batalin-Vilkovisky (BV) formalism~\cite{Batalin:1981jr,Batalin:1983wj} has proved a powerful tool not only in quantization but more generally in analyzing physical content of gauge systems and even constructing new gauge models. In the context of local gauge theories the appropriate enhancement~\cite{Barnich:1995db,Barnich:1995ap,Barnich:2000zw} (see also \cite{Stora:1983ct,DuboisViolette:1985jb,Piguet:1995er}  for earlier important developments) of the BV approach operates in terms of the jet-bundles associated to fields, ghost fields and antifields.

If one is only interested in the equations of motion and hence disregards Lagrangians and associated graded symplectic structures the appropriate version~\cite{Barnich:2004cr} (see also~\cite{Lyakhovich:2004xd,Kazinski:2005eb}) of the local BV formalism can be formulated immediately in terms of manifolds which are more general than jet-bundles, giving a powerful generalization~\cite{Barnich:2010sw,Grigoriev:2019ojp} of the standard approach and leading to more invariant and flexible description of gauge systems.

More precisely, a generic local gauge field theory formulated at the level of equations of motion can be represented~\cite{Barnich:2010sw} as a nonlagrangian version~\cite{Barnich:2005bn} of AKSZ-type sigma model, whose target space is a BV jet-bundle of the system or one of its equivalent reductions. In this way one can define and analyze local gauge field theories in terms of generic $Q$-manifolds that are not necessarily jet-bundles. This approach can be regarded as a BV extension of the invariant geometrical approach to PDEs~\cite{Vinogradov1981} (see also~\cite{vinogradov:2001,Krasil'shchik:2010ij}). It can also be considered as an extension of the AKSZ construction~\cite{Alexandrov:1995kv} to the case of not necessarily topological gauge theories. 

At the Lagrangian level a generic local gauge system can be also represented~\cite{Grigoriev:2010ic,Grigoriev:2012xg} as an AKSZ-type sigma model whose target space is a suitably defined graded cotangent bundle over the jet-bundle associated to fields and ghosts or one of its equivalent reductions.  This approach also known as the Lagrangian parent formulation has certain remarkable features. In particular, just like conventional AKSZ sigma models it automatically contains  Batalin-Fradkin-Vilkovisky (BFV)~\cite{Batalin:1983pz,Batalin:1977pb} Hamiltonian formulation. Moreover, the approach gives a systematic way to derive frame-like description of the system: for instance in the case of Einstein gravity the familiar frame-like (also known as Cartan-Weyl) formulation~\cite{Weyl:1929} in terms of the frame-field and Lorentz connection can be systematically arrived at~\cite{Grigoriev:2010ic} as a suitable equivalent reduction of the parent formulation for the metric-like formulation of Einstein gravity.

Despite its nice supergeometrical structure, the Lagrangian parent formulation involves an infinite tower of generalized auxiliary fields.  It turns out that by eliminating most of auxiliaries but at the same time trying to keep the supergeometrical structure intact one can naturally arrive at so-called presymplectic AKSZ formulations~\cite{Alkalaev:2013hta}. These have the form of finite-dimensional AKSZ sigma models whose target space presymplectic structure is allowed to be degenerate. In this way one can find elegant presymplectic AKSZ formulations~\cite{Alkalaev:2013hta} for a variety of gauge theories including the frame-like form of Einstein gravity. It turns out that the ghost-independent part of the presymplectic AKSZ action for gravity is precisely the Cartan-Weyl action. However, if space-time dimension is greater than 3 the BV-like 2-form defined on the space of supermaps to the target space is degenerate and the BV interpretation of such a presymplectic AKSZ sigma model has remained somewhat unclear. 

Later on it was realised~\cite{Grigoriev:2016wmk} that the presymplectic 2-form on the target space is closely related to a BV extension of the canonical 2-form on the stationary surface, which is induced by the Lagrangian. As we are going to see it can also be seen as a BV symplectic structure completed to a cocycle of the total differential $\dh+s$ and transferred to the minimal formulation of the BV-BRST complex.\footnote{The descent symplectic structures completing the BV symplectic structure have been discussed in~\cite{Sharapov:2016sgx}.} Here, $\dh$ denotes the horizontal differential and $s$ the BV-BRST differential of the theory.

In this work we concentrate on the example of Einstein Gravity and give a consistent interpretation of its presymplectic AKSZ formulation. More specifically, we demonstrate that the BV presymplectic 2-form on the space of supermaps is regular and factoring out its kernel results in the standard symplectic BV field-antifield configuration space while the AKSZ action functional induces the BV master action, giving  a concise and geometrical BV formulation of the frame-like gravity. Analogous procedure applied to the presymplectic AKSZ sigma model restricted to the spatial slice of the space-time results in the BFV phase space and 1-st class constraints of the frame-like gravity.   However the presymplectic structure in this case is not regular and the phase space is recovered as a maximal symplectic submanifold of the respective space of supermaps.  Up to this subtlety, just like in the case of usual AKSZ, its presymplectic version also contains both BV and BFV formulations. Note that the BFV  phase-space encoded in the proposed presymplectic formulation is precisely the one of~\cite{Canepa:2020rhu,Canepa:2020ujx}, where the relation between BV and BFV for the frame-like gravity has been recently studied.

\section{Presymplectic AKSZ form of Cartan-Weyl action}
\label{sec:psymp}

Let $\algg[1]$ be a linear space of Poincar\'e or (A)dS algebra in $n$ dimensions with the degree shifted by 1 and regarded as a graded manifold. The standard coordinates  $\xi^a,\rho^{ab}$ are Grassmann odd variables of degree $1$ associated to the translation (transvections) and the Lorentz rotation generators respectively. The Lie algebra structure on $\algg$ defines a $Q$-structure on $\algg[1]$, which can be identified with the Chevalley-Eilenberg (CE) differential of $\algg$. In terms of the coordinates it is given by
\begin{equation}
q \xi^a=\rho^a{}_b \xi^b\,, \qquad
q \rho^{ab}=\rho^{a}{}_{c}\rho^{cb}+\lambda \xi^a\xi^b,
\end{equation}
where parameter $\lambda$ is related to the cosmological constant through $\lambda=-\frac{2\Lambda}{(n-1)(n-2))}$. At $\lambda=0$ this gives CE differential of the Poincar\'e algebra. 

On $\algg[1]$ there is a natural $q$-invariant presymplectic structure of degree $n-1$, which reads as~\cite{Alkalaev:2013hta}:
\begin{equation}
    \omega^{\algg[1]}= \Vol_{abc}(\xi) d\xi^{a}d\rho^{bc}\,,\qquad
    \Vol_{a_1\ldots a_k}(\xi)=\frac{1}{(n-k)!}\epsilon_{a_1\ldots a_{k}b_{1}\ldots b_{n-k}} \xi^{b_1}\ldots \xi^{b_{n-k}}\,.
\end{equation}

It follows from $L_q\omega^{\algg[1]}=0$ and 
$d \omega^{\algg[1]}=0$ that $d\ii_q\omega^{\algg[1]}=0$ and hence
\begin{equation}
    \ii_q \omega^{\algg[1]} =d\cH   \,, \qquad \ii_q\ii_q\omega^{\algg[1]}=0=q\cH\,, 
\end{equation}
for some function $\cH$ of degree $n$. Furthermore, there exists a presymplectic potential $\chi\in \Lambda^{1}(\algg[1])$ such that $\omega^{\algg[1]}=d\chi$. In the case at hand one can take:
\begin{equation}\label{chi-H}
\chi= \Vol_{ab}(\xi) d\rho^{ab} \,, \qquad \cH=
 \Vol_{ab}(\xi)(\rho^{a}{}_c\rho^{c b}+\frac{(n-2)\lambda }{n}\xi^a\xi^b)\,.
\end{equation}
Note that in contrast to $\chi$, which is defined up to a $d$-closed 1-form, function $\cH$ is uniquely determined by $\omega^{\algg[1]}$ and $q$.

Let us consider the presymplectic AKSZ sigma model with the source space being $(T[1]X,d_X)$, where $X$ is a space-time manifold of dimension $n>2$ and $d_X$ is the de Rham differential seen as a homological vector field on $T[1]X$, and the target space being $(\algg[1],q,\omega^{\algg[1]})$. Maps from $(T[1]X,d_X)$ to $(\algg[1],q,\omega^{\algg[1]})$ are field configurations of the Cartan-Weyl formulation of gravity. Indeed, in terms of coordinates a map $\sigma$ is parameterized by 
\begin{equation}
\label{map-param}
    \sigma^*(\xi^a)=e^a_\mu(x) \theta^\mu\,, \qquad 
    \sigma^*(\rho^{ab})=\omega^{ab}_\mu(x) \theta^\mu\,,
\end{equation}
where $x^\mu,\theta^\mu$ are standard local coordinates on $T[1]X$ induced by local coordinate son the base $X$. Fields
$e^{a}_\mu(x),\omega^{ab}_\mu(x)$ are identified with the usual frame field and Lorentz connection and, as usual, we require configurations to be such that $e^a_\mu(x)$ is invertible. 

The data of AKSZ sigma-model determine the action functional as follows:
\begin{equation}
\label{AKSZ-action}
S[\sigma]=\int_{T[1]X}( \sigma^*(\chi)(d_X)+\sigma^*(\cH))
\end{equation}
where $\chi\in \Lambda^{1}(\algg[1])$ and $\cH \in \Lambda^{0}(\algg[1])$ are given in~\eqref{chi-H}. In component one has:
\begin{multline}
\label{CW-action}
  S[e,\omega]=\int_X \Vol_{ab}(e) (d_X \omega^{ab}+\omega^{a}{}_c\omega^{c b}+\frac{(n-2)\lambda}{n}e^ae^b)=\\
  \int_X \Vol_{ab}(e) (d_X \omega^{ab}+\omega^{a}{}_c\omega^{c b})-2\Lambda \Vol(e)
\end{multline}
where fields $e^a,\omega^{ab}$ parameterize $\sigma$ according to \eqref{map-param} and the wedge product of space-time differential forms is assumed. This is indeed a familiar Cartan-Weyl action. 

\section{BV-AKSZ interpretation of the model}
If $\omega$ is nondegenerate (which is not the case for gravity in $n>3$) the BV formulation is extracted from the AKSZ data as follows: the BV field-antifield space is the space of supermaps from $T[1]X$ to the target supermanifold.  In contrast to the space of maps the space of supermaps is a graded manifold. The coordinates there can be introduced as follows
\begin{gather}
    \hat\sigma^*(\xi^a)=\st{0}{\xi}{}^a(x)+e^a_\mu(x) \theta^\mu+\half\st{2}{\xi}{}^a_{\mu\nu}(x)\theta^\mu\theta^\nu+\ldots\,, \\
    \hat\sigma^*(\rho^{ab})=\st{0}{\rho}{}^{ab}(x)\omega^{ab}_\mu(x) \theta^\mu+\half\st{2}{\rho}{}^{ab}_{\mu\nu}(x)\theta^\mu\theta^\nu+\ldots\,,
\end{gather}
where now the form-degree $k$ components carry ghost degree $1-k$. The space of maps is recovered by setting to zero all the coordinates of nonvanishing degree.

The target space symplectic structure $\omega^{\algg[1]}$ determines a BV symplectic structure of degree $-1$ on the space of supermaps:
\begin{equation}
\label{symp}
\omega^{BV}=\int_{T[1]X}\hat\sigma^*(\omega^{\algg[1]}_{AB})\delta \psi^A(x,\theta)\wedge  \delta \psi^B(x,\theta)
\end{equation}
where we introduced a collective notation $\psi^A$ for target space coordinates $\xi^a,\rho^{ab}$ and $\psi^A(\theta)=\hat\sigma^*(\psi^A)$.  It is nondegenerate provided $\omega$ is. The BV action is given by
\begin{equation}
\label{BV-AKSZ-action}
S^{BV}[\hat\sigma]=\int_{T[1]X}( \hat\sigma^*(\chi)(d_X)+\hat\sigma^*(\cH))\,,
\end{equation}
where the difference with~\eqref{AKSZ-action} is in $\sigma$ replaced by $\hat\sigma$. In particular, setting fields of nonzero degree to zero one recovers~\eqref{AKSZ-action}. 
For further details and developments of the AKSZ approach we refer to~\cite{Alexandrov:1995kv,Cattaneo:1999fm,Grigoriev:1999qz,Batalin:2001fc, Park:2000au,Roytenberg:2002nu,Barnich:2005ru,Kazinski:2005eb,Kotov:2007nr, Bonechi:2009kx,Barnich:2009jy,Cattaneo:2012qu,Grigoriev:2012xg,Boulanger:2012bj,Ikeda:2012pv,Bonavolonta:2013mza}.

In the presymplectic case one can still define
the BV-like action \eqref{BV-AKSZ-action}
and the presymplectic structure~\eqref{symp} on the space of supermaps.  These satisfy an analog of the master equation that can be defined as follows: the homological vector field $q$ in the target space and the de Rham differential $d_X$ on $X$ naturally define the BRST differential $s$ on the space of supermaps \cite{Alexandrov:1995kv}:
\begin{equation}
\label{s-AKSZ}
s=\int d^nx d^n\theta (d_X \psi^A(x,\theta)+q^A(\psi(x,\theta))\vdl{\psi^A(x,\theta)}\,.
\end{equation}
One can then check that by construction
\begin{equation}
\label{master-2}
\omega^{BV}(s,s)=0\quad \qquad \ii_s \omega^{BV}=\delta S^{BV}\,,
\end{equation}
modulo boundary terms.

As we are going to see the presymplectic structure is  regular in a certain precise sense and hence it defines the symplectic structure on the symplectic quotient space, i.e. the space of leaves of the kernel distribution determined by $\omega^{BV}$. One can then check that $S^{BV}$ is annihilated (modulo boundary terms) by the distribution and hence defines a well defined functional on the space of leaves.  Moreover, in a similar way both $s$ and $\omega^{BV}$ induces the respective structures on the quotient and altogether they satisfy the analog of~\eqref{master-2}. Finally, because the presymplectic structure induced on the quotient is nondegenerate ~\eqref{master-2} implies usual BV master equation on the symplectic quotient and hence this data defines a conventional BV formulation on the symplectic quotient. 

Furthermore, one can check that $e^a_\mu$ and $\omega^{ab}_\mu$ are not in the kernel of $\omega^{BV}$ so that the ghost-independent part of the BV action on the quotient space is just the Cartan-Weyl action~\eqref{CW-action}. Together with the facts that the symplectic structure is nondegenerate on the quotient, the spectrum of ghost fields precisely corresponds to the gauge invariance of the action, and the master equation holds this implies that we have indeed arrived at the BV formulation of gravity in the Cartan-Weyl form.

Because the above consideration deals with infinite-dimensional manifolds some care is required. However, as we are going to see in the next section the factorization boils down to that of the finite-dimensional manifold while the construction of basic objects can be made precise by employing the jet-bundle technique.

Moreover, it turns out the symplectic quotient can be explicitly realised as a submanifold $J_X(N)$ of the entire manifold $J_X(M)$ of supermaps from $T[1]X$ to $\algg[1]$, which is transversal to the kernel distribution. In this way $S^{BV},\omega^{BV}$ on the quotient can be obtained by simply pulling back these structures to $J_X(N)$. 

It is important to stress that in order to study the theory there is no need to explicitly identify the symplectic quotient. The master equation, gauge fixing etc. can be implemented just in terms of $J_X(M)$.
Moreover, constraints determining the submanifold do not involve space-time derivatives so that they can be easily implemented e.g. in the path-integral and hence, at least formally, quantization can be also performed without explicit restriction to the symplectic quotient.  The formalism we have arrive at can be regarded as a presymplectic BV-AKSZ formalism or a version of BV-AKSZ formalism with constraints.

\subsection{The structure of the fiber and its symplectic quotient}

The space of (super)maps from $T[1]X$ to $\algg[1]$ can be locally represented as the space of (super)maps from $X$ to $M$, where $M$ is a space of (super)maps from $T_x[1]X$ to $\algg[1]$ for a given $x\in X$. It is clear that $M$ is finite-dimensional. If $\psi^A$ and $\theta^\mu$ are coordinates on $\algg[1]$ and $T_x[1]X$ respectively then a generic supermap is determined by a function $\psi^A(\theta)$ whose coefficients  can be taken as coordinates on $M$. It is convenient to employ $\psi^A(\theta)$ as a generating function for coordinates on $M$.

In the case at hand we chose $M$ to be the space of supermaps from $T_x[1]X$ to $\algg[1]$ satisfying the additional condition that the component $e^a_\mu$ entering $\xi^a(\theta)$ as $e^a_\mu \theta^\mu$ is required to be nondegenerate. Speaking geometrically $M$ is the fiber bundle over $GL(n,\fR)$.

The presymplectic structure on $\algg[1]$ determines that on $M$ via
\begin{equation}
    \omega^M=\int d^n\theta\,\, \omega^{\algg[1]}_{AB}(\psi(\theta))d\psi^A(\theta)
    \wedge d\psi^B(\theta)\,,
\end{equation}
where $d\psi^A(\theta)$ is a generating function for basis differentials of coordinates on $M$. To analyse the structure of $\omega^M$ it is convenient to
consider a submanifold $M_0\subset M$ determined by
\begin{equation}
{\xi^a}=0,\qquad \st{2}{\xi}{}^a_{\mu\nu}=0, \qquad 
\ldots\,,\qquad
\st{n}{\xi}{}^a_{\mu_1\ldots \mu_n}=0    \,,
\end{equation}
where the coordinate functions $\st{l}{\xi}{}^a_{\mu_1\ldots \mu_l}$ are introduced as follows:
\begin{equation}
    \xi^a(\theta)=\xi^a+e^a_\mu \theta^\mu+\half \st{2}{\xi}{}^a_{\mu\nu}\theta^\mu\theta^\nu
+\ldots
+\frac{1}{n!}\st{n}{\xi}{}^a_{\mu_1\ldots \mu_n}\theta^{\mu_1}\ldots \theta^{\mu_n}\,,
\end{equation}
and by some abuse of notations we denote $\st{0}{\xi}{}^a$ by $\xi^a$. Together with coordinate functions $\rho^{ab},\omega_\mu^{ab},\st{i}{\rho}{}^{ab}_{\mu_1\ldots\mu_i}$, $i=2,\ldots,n$ introduced in a similar way these provide a natural coordinate system on $M$. 

To understand the structure of the kernel of $\omega^M$ it is instructive to consider $\omega^M$ at a given point $p\subset M_0$. By changing the basis in $T_x X$ one can assume that $e^a_\mu=\delta^a_\mu$ to further simplify the analysis. In this basis the explicit expression for $\omega^M$ at $p$ reads as
\begin{equation}
\label{omega-p}
\omega^M_p = 
de^b_c\wedge d\st{2}{\rho}{}^{c}_{b}+
d\xi^b \wedge d \st{3}{\rho}{}_b+
d\rho^{ab}\wedge d\st{3}{\xi}{}_{ab}+
d\omega^{ab}_c \wedge d\st{2}{\xi}{}^c_{ab}\,,
\end{equation}
where $\st{2}{\rho}{}^{c}_{b}$, $\st{3}{\rho}{}_b$, $\st{3}{\xi}{}_{ab}$, and $\st{2}{\xi}{}^c_{ab}$ parameterize the following components:
\begin{equation}
\st{2}{\rho}{}^{cd}_{bd}\,, \qquad 
\st{3}{\rho}{}^{cd}_{bcd}\,, \qquad 
\st{3}{\xi}{}^{c}_{abc}\,, \qquad 
\st{2}{\xi}{}^c_{ab}\,.
\end{equation}
The remaining components are in the kernel of $\omega^M_p$. It is easy to see that the spectrum of coordinates along which $\omega^M$ is nondegenerate is precisely that required for  minimal BV formulation of GR. Let us stress that for the moment this is only established at $M_0$.

The crucial fact is that $M$ is a regular presymplectic manifold. To see this consider the following vector fields on $\algg[1]$ (here and below we restrict to 4d to simplify the analysis):
\begin{equation}
\label{VF}
\begin{gathered}
X^4_a=\xi^{(4)}\dl{\xi^a}\,, \qquad X^3_{ab}=\xi^{(3)}_a\dl{\xi^b}+\xi^{(3)}_b\dl{\xi^a}\\ Y^4_{ab} =\xi^{(4)}\dl{\rho^{ab}}\,, \quad Y^{3}_{abc}=\xi^{(3)}_a\dl{\rho^{bc}}+\xi^{(3)}_b\dl{\rho^{ac}}\,, \quad 
Y^{2}_{abcd}=\xi^{(2)}_{ab}\dl{\rho^{cd}}+\ldots\,,
\end{gathered}
\end{equation}
where $\ldots$ in the last expression denotes terms symmetryzing the expression in $ac$ and $bd$. Here $\xi^{(k)}_{a_1\ldots a_{4-k}}$ denote $\frac{1}{k!}\epsilon_{a_1\ldots a_{4-k}a_{4-k+1}\ldots a_4}\xi^{a_{4-k+1}}\ldots \xi^{a_4}$. It is easy to check that all these vector fields are in the kernel of $\omega^{\algg[1]}$ and commute to one another.

By natural prolongation vector fields~\eqref{VF} on $\algg[1]$ determine the vector fields on $M$. Given a vector field $X$ on $\algg[1]$ the component expression for its prolongation $\hat X$ can be obtained as follows:
\begin{equation}
    \hat X\psi^A(\theta)=X^A(\psi(\theta))\,,
\end{equation}
where $X^A=X\psi^A$. Prolongation commutes with the commutator. In particular, prolongations of an involutive set of vector fields on $\algg[1]$ is again an involutive set. What is less trivial is that the distribution determined by an involutive set on $\algg[1]$ can be  nonregular while
its prolongation is regular. This happens because in our case $M$ is not the space of all maps but only of those whose $e^a_\mu$ component is invertible. 

We have the following:
\begin{lemma} 
\label{lemma:reg}
The distribution determined by the prolongations of vector fields~\eqref{VF} coincides with the kernel distribution of $\omega^M$ on $M$ and hence $(M,\omega^M)$ is a regular presymplectic manifold. 
\end{lemma}
\begin{proof}
First we show that on $M_0$ these vector fields exhaust the kernel of $\omega^M$. This is easy to believe because the fields are linearly independent and their tensor structure precisely corresponds to the kernel of $\omega^M$ at $M_0$, cf.~\eqref{omega-p}. The proof is purely technical and is relegated to the Appendix~\bref{app:VF-kernel}.

Next, by construction vector fields $\hat X,\hat Y$ are in the kernel of $\omega^M$ everywhere while at $M_0$ they exhaust the kernel. It follows they define the kernel everywhere. Indeed the dimension of the distribution determined by $\hat X,\hat Y$ can't drop when moving off $M_0$ (because $M$ is a formal neighbourhood of $M_0$). At the same time the rank of $\omega^M$ can't drop off $M_0$ as well so that the rank must be constant.\footnote{Another way to see that is to fix the concrete values of $e^a_\mu,\omega_\mu^{ab}$ and consider $\omega^M$ as function of the remaining coordinates. Then it has the form~\eqref{omega-p} plus terms proportional to the remaining coordinates of nonvanishing degree. But such terms can't decrease the rank. The same argument applies to the distribution determined by the vector fields.}
\end{proof}

The above proof of regularity of $\omega^M$ employs vector fields~\eqref{VF} that we explicitly gave only for the case of $n=4$. It turns out the proof can be extended to generic $n>4$ as follows. Observe that the linear space (over $\fR$) of vector fields on $\algg[1]$ with $\rho$-independent coefficients is isomorphic to the tangent space $T_p M$ at $p\in M_0$. The image of a given vector field on $\algg[1]$ is determined by its prolongation to $M$ considered at $p$. In a similar way $T^*_p M$ is isomorphic to the space of $1$-forms on $\algg[1]$ with
$\rho$-independent coefficients. Moreover these isomorphisms are compatible with the map from vector fields (tangent vectors) to 1-forms (resp. $T^*_pM$) determined by $\omega^{\algg[1]}$ (resp. $\omega^{M}|_p$). This implies that the kernel of $\omega^{M}|_p$ is determined by prolongations of target space vector fields from the kernel of $\omega^{\algg[1]}$ so that the arguments given in the proof of Lemma~\bref{lemma:reg} imply regularity of $\omega^M$.

The regularity of $(M,\omega^M)$ implies that there exists (at least locally) a symplectic quotient $N$ of $M$.  In particular, by Frobenius theorem one can introduce a coordinate system $z^i,w^\alpha$ such that the vector fields $\hat X, \hat Y $ take the form $\dl{{w^\alpha}}$. The submanifold singled out by $w^\alpha=0$ and equipped with a pullback of $\omega^M$ is isomorphic (as a symplectic manifold) to the symplectic quotient. 

However, the Frobenius coordinates are not so easy to find explicitly.  Nevertheless any submanifold transversal to the distribution determined by $\hat X,\hat Y$ and equipped with the induced symplectic structure is also isomorphic (as a symplectic manifold) to the symplectic quotient. Such transversal manifolds can be easily found using the following:
\begin{lemma}
\label{lemma:w-sub}
Let $M$ be a graded presymplectic manifold and $M_0\subset M$ be its submanifold determined by equations $k_r=0$, where $\gh{k_r}\neq 0$. Let $z^i,w^\alpha$ be local homogeneous (i.e. of definite degree) coordinates on $M$ such that $\omega^M(\dl{w^\alpha},\cdot)=0$ on $M_0$, $\dl{w^\alpha}$ determine a kernel of $\omega^M$ at each point of $M_0$, and $\gh{w^\alpha}\neq 0$. Then (in general, locally defined) submanifold $N$ determined by $w^\alpha=0$ is symplectic. Moreover, $N$ is (locally) symplectomorphic to the symplectic quotient of $M$, provided $\omega^M$ is regular.
\end{lemma}
\begin{proof}
The pullback  $\omega_N$ of $\omega$ to $N$ is by construction nondegenerate at each point of $N_0=N\cap M_0$. The standard considerations then ensure that $\omega_N$ can't degenerate off  $N_0$. This shows that $N$ is symplectic. 
\end{proof}

Using the above Lemma it is not difficult to find a convenient choice of such functions $w^\alpha$. For instance, consider e.g.:
\begin{equation}
\st{4}{\xi}{}^{a|}\,,\quad 
\st{4}{\rho}{}^{ab|}\,, \quad
\st{3}{\xi}{}^{a|\mu}e_\mu^b+(ab)\,,\quad
\st{3}{\rho}{}^{ab|\mu}e_\mu^c+(ac)\,,\quad
\st{2}{\rho}{}^{ab|\mu\nu}e_\mu^c e_\nu^d+(ac)(bd)\,,
\end{equation}
where e.g. $\st{3}{\xi}{}^{a|\mu}$ stand for $\st{3}{\xi}{}^{a}_{\nu\rho\sigma}\epsilon^{\mu\nu\rho\sigma}$.
Because this set is manifestly $o(3,1)$ invariant it can be useful in practical computations though in this work we do not make use of these functions. Note the first 2 constraints originate from the target space in the sense that they can be represented as $\int d^4\theta\xi^a(\theta)$ and $\int d^4\theta \rho^{ab}(\theta)$. What is less trivial is that the last two can be replaced with those originating from the target space. More precisely, on $M_0$ they coincide with:
\begin{equation}
\int d^4\theta \rho^{ab}(\theta)\xi^c(\theta)+(ac)\,,
\qquad    
\int d^4\theta \rho^{ab}(\theta)\xi^c(\theta)\xi^d(\theta)+(ac)(bd)\,.
\end{equation}
This property could simplify implementation of these constraints in applications.

Given a regular presymplectic manifold $M$  we have the following:
\begin{lemma} Let $Q$ be a homological vector field on $M$ satisfying $\ii_Q\omega^M=d\cH^M$ and $\omega^M(Q,Q)=0$ for some function $\cH^M$. It follows:
\begin{equation}
\ii_{Q_N}\omega^N=d\cH^N \,, \qquad \omega^N(Q_N,Q_N)=0
\end{equation}
where $Q_N,\cH_N$ and $\omega^N$ are induced by $Q,\cH^M$ and $\omega^M$ respectively on the symplectic quotient $N$. If $N$ is identified as the surface $N\subset M$ then $\cH^N,\omega^N$ also coincides with $\cH^M,\omega^M$ pulled back to $N$.
\end{lemma}
\begin{proof}
The first part of the statement is standard and can be e.g. easily seen using special coordinates $z^i,w^\alpha$. That $\cH_N,\omega^N$ coincides with $\cH^M,\omega^M$ pulled back to $N$ is true because $\cH^M,\omega^M$ are constant along the kernel of $\omega^M$. Indeed, from $\ii_X\omega^M=0$ one finds $L_X\omega^M=0$ and $X\cH=\ii_X d\cH=\ii_X \ii_Q \omega=0$.
\end{proof}
Note that the statement remains true if instead of regularity one requires $N\subset M$ to be symplectic and such that $T_p M=T_p N\oplus \ker(\omega^M_p)$ for any $p\in N$.

Let us discuss the structure of $N$ and natural coordinates therein. First of all one observes that coordinates $\xi^a,\rho^{ab}, e_\mu^a, \omega_\mu^{ab}$ remain independent when restricted to $N$ and hence give a part of the coordinate system on $N$. The remaining coordinates are of negative ghost degree so that there is an invariantly defined submanifold $N_{01} \subset N$ obtained by setting them to zero. 
\begin{lemma}
$N_{01}$ is a Lagrangian submanifold of $N$. $N$ can be identified as $T^*[-1]N_{01}$.
\end{lemma}
\begin{proof}
Using standard coordinates on $M$ one finds that
\begin{equation}
\omega^M(\dl{\xi^a},\dl{\rho^{bc}})=
\epsilon_{abcd}(\st{4}{\xi}{}^a_{\mu\nu\rho\sigma}\epsilon^{\mu\nu \rho\sigma})    
\end{equation}
so that indeed it vanishes when $\st{4}{\xi}{}^a_{\mu\nu\rho\sigma}$ (which is of degree $-3$) is set to zero. In a similar way one finds that $\omega$ vanishes on all pairs of vectors tangent to $N_{01}$. 
\end{proof}
It follows from the above Lemma that identifying
$N$ as a $T^*[-1]N_{01}$ it can be convenenient to use Darboux coordinates given by $\xi^a,\rho^{ab}, e_\mu^a, \omega_\mu^{ab}$ and their canonically conjuated antfields:
\begin{equation}
    \omega^N =d\xi^a \wedge d \xi^*_a+d \rho^{ab}\wedge d\rho^*_{ab}+d e_\mu^a \wedge d e^*{}^\mu_a+ d \omega_\mu^{ab}\wedge  
    d \omega^*{}_\mu^{ab}\,.
\end{equation}
Promoting these variables to fields on $X$ gives a standard set of fields, ghosts and their conjugated antifields required for the BV formulation of the Cartan-Weyl GR.

\subsection{BV from presymplectic AKSZ}

Let us now  turn to the piece of the AKSZ action determined by the De Rham differential. To this end consider the jet-bundle $J_X(M)$ associated to  $X \times M \to X$ and denote by $D_\mu$ the total derivatives with respect to $x^\mu$. It is useful to identify vertical coordinates on $J_X(M)$ as coefficients of the generating functions $\psi^A(y,\theta)$ which are formal power series in auxiliary coordinates $y^\mu$ and degree $1$ coordinates $\theta^\mu$. Here $\psi^A$ are coordinates on $\algg[1]$. In this representation it is clear that $J_X(M)$ can also be defined as a super jet-bundle associated to $T[1]X\times \algg[1]\to T[1]X$.  

On $J_X(M)$ we define the homological vector field $D$ via its action on coordinates:
\begin{equation}
D \psi^A(y,\theta)=\theta^\mu \dl{y^\mu} \psi^A(y,\theta)    \,, \qquad Dx^\mu=0
\end{equation}
Another useful representation for $D$ is as follows:
$D \psi^A(y,\theta)=\theta^\mu D_\mu \psi^A(y,\theta)$\,.

The symplectic structure on $M$ defines a symplectic structure on $J_X(M)$:
\begin{equation}
\omega^{J(M)}=\int d^n\theta (\omega^{\algg[1]}_{AB}(\psi(\theta))\dv \psi^A(\theta)\dv \psi^B(\theta)=
\omega_{ij}\dv z^i \dv z^j
\end{equation}
where in the last equality we used special coordinate system  $z^i,w^\alpha$ on $M$, such that  $\omega(\dl{w^\alpha},\dl{w^\beta})=0=
\omega(\dl{w^\alpha},\dl{z^i})$ on $N$, and whose existence has been proved in the previous section. More precisely, $z^i,w^\alpha$ denote coordinates on $J_X(M)$ obtained by pulling back $z^i$ to $J_X(M)$.

Understood as a local function on $J_X(M)$ the integrand (over $X$) of the AKSZ master-action \eqref{BV-AKSZ-action} takes the following form:
\begin{equation}
\begin{gathered}
L^{BV}=K-\cH_M\,, \\
 K= \int d^n\theta \chi_A(\psi(\theta)) \theta^\mu D_\mu \psi^A(\theta)
\,,
\qquad \cH^M=\int d^n\theta \cH(\psi(\theta))\,,
\end{gathered}
\end{equation}
where $\chi_A$ are components of the presymplectic potential on $\algg[1]$ and we identify functions on $M$ and their pullbacks to $J_X(M)$.

In the jet-bundle terms the BRST differential $s$ (cf.~\eqref{s-AKSZ}) is represented by a vertical evolutionary vector field $s$:
\begin{equation}
s=D+Q^{pr}    
\end{equation}
where $Q^{pr}$ is the prolongation of $Q$ to $J_X(M)$ determined by $\commut{Q^{pr}}{D_\mu}=0$. Furthermore, one has the following relations 
\begin{equation}
\label{iDomega}
    \ii_{D}\omega^{J(M)}=\dv K+D_\mu (\cdot)^\mu\,, \qquad  \ii_{D} \ii_{D} \omega^{J(M)}=D_\mu (\cdot)^\mu
\end{equation}
as well as
\begin{equation}
\qquad \ii_{D} \ii_{Q^{pr}} \omega^{J(M)}=D_\mu (\cdot)^\mu\,.
\end{equation}
which can be directly checked and amount to:
\begin{equation}
\label{pme}
    \ii_s \ii_s \omega^{J(M)}=\dv L^{BV}+D_\mu (\cdot)^\mu\,, \qquad  \ii_{s} \ii_{s} \omega^{J(M)}=D_\mu (\cdot)^\mu\,,
\end{equation}
where we also took into account $\ii_Q\ii_Q \omega^M=0$.

Let us now consider the jet sub-bundle $J_X(N) \subset J_X(M)$ determined by $w^\alpha=0$ and their  prolongations. Note that $J_X(N)$ can be also seen as the jet-bundle associated to $X\times N \to X$. Restricting relations~\eqref{pme} to $J_X(N)$ one gets:
\begin{equation}
\label{pme2}
\ii_{s^N} \omega^{J(N)}=\dv L^{BV}_N+D_\mu (\cdot)^\mu\,, \qquad \ii_{s^N} \ii_{s^N} \omega^{J(N)}=D_\mu (\cdot)^\mu
\end{equation}
where $s^N$ denotes the projection of $s$ from $J_X(M)$ to  $J_X(N)$ (induced by the projection $M\to N$) and $\omega^{J(N)}, L^{BV}_N$ denote the respective objects on $J_X(M)$ pulled back to to $J_X(N)$. Here by some abuse of notations $D_\mu$ denotes a total derivative on either $J_X(M)$ of $J_X(N)$; this does not lead to confusions because $D_\mu$ are tangent to $J_X(N)\subset J_X(M)$. The above relations are obvious if one makes use of the special coordinate system on $J_X(M)$ induced by special coordinates $z^i,w^\alpha$ on $M$.


Because $N$ is symplectic we are dealing with the standard BV formulation so that the above can be rewritten as $\ab{L_N^{BV}}{L_N^{BV}}=D_\mu (\cdot)^\mu$, where  $L_N^{BV}$
is $L^{BV}$ restricted to $J_X(N)$. To summarize: understood as a local function on $J_X(M)$ the AKSZ action \eqref{BV-AKSZ-action} is precisely $L^{BV}=K+\cH_M$. The sub-bundle $J_X(N)$ obtained by factoring out the kernel of the symplectic structure on $M$ is a BV jet-bundle equipped with the BV symplectic structure. The restriction of $K+\cH$ to $J_X(N)$ satisfies master-equation and hence determines the BV formulation of the frame-like gravity. This is indeed true because (i) $N$ is symplectic (ii) $L^{BV}$ restricted to the body of $J_X(N)$ is just the Cartan-Weyl Lagrangian (iii) the terms linear in ghosts and antifields contains the complete set of gauge generators. Note that the resulting BV formulation on $J_X(N)$ does not depend on how exactly $N$ is realised as a submanifold of $M$. 

The only point that requires clarification is (iii) because this ensures that $L^{BV}$ is a proper solution to the master equations.
To see this let us spell out explicitly the terms in $L^{BV}$ that are linear in ghosts and linear in antifields (i.e. in $\xi^a$, $\rho^{ab}$ and  $\st{2}{\xi}{}^{a}_{\mu\nu}$, $\st{2}{\rho}{}^{ab}_{\mu\nu}$):
\begin{equation}
\label{gs}
\begin{gathered}
    \epsilon_{abcd}\st{2}{\rho}{}^{ab}((\nabla\xi)^c e^d -\xi^c(\nabla e)^d)\,,\qquad
    \epsilon_{abcd}\st{2}{\rho}{}^{ab} \rho^{c}_{e}e^e e^d\,,\\
    \epsilon_{abcd}\st{2}{\xi}{}^{a} \xi^b (d_x\omega+\omega\omega)^{cd}\,,\qquad
    \epsilon_{abcd}\st{2}{\xi}{}^{a}  (\nabla\rho)^{bc}e^d \,.
    \end{gathered}
\end{equation}
Here the integration over $d^n\theta$ is left implicit, $\nabla$ denotes a covariant differential with respect to $\omega^{ab}$, and for simplicity we set $\lambda=0$.  The first line encodes gauge transformations of $e^a_\mu$  and the second line encodes the gauge transformations of $\omega^{ab}_\mu$ (note that $\st{2}{\xi}{}^{a}_{\mu\nu}$ parameterize antifields conjugate to $\omega^{ab}_\mu$) provided one identifies $\xi^a$ and $\rho^{ab}$ as parameters (in a certain basis) of the diffeomorphisms and the local Lorentz transformations respectively.

However, it is difficult to explicitly compare $L^{BV}$ with the standard expression~\cite{Barnich:1995ap} of the BV master action for Cartan-Weyl action because even the above terms involve $\st{2}{\rho}{}^{ab}_{\mu\nu}$
which give an overcomplete set of coordinates in this sector and moreover the symplectic structure is not in the canonical Darboux form. Fortunately, in order to prove that $L^{BV}$ is a proper BV master action all we need to demonstrate is that it is proper (i.e. all gauge generators are taken into account). It is enough to do so in quadratic approximation because nonlinear correction can't bring extra degeneracy. 

Let us analyze the linearization of the gauge symmetries encoded in~\eqref{gs} around the vacuum solution $e^a_\mu=\delta^a_\mu$ and $\omega^{ab}_\mu=0$. One finds:
\begin{equation}
    \delta \bar\omega^{ab}_{c}=\d_c \rho^{ab}\,, \qquad \delta \bar e^a_b=\rho^a{}_b+\d_b\xi^a\,,
\end{equation}
where $\bar e^a_b$ and $\bar \omega^{ab}_c$ are related to the linearized $e^a_b$ and $\omega^{ab}_c$ respectively through a linear invertible redefinition. The above are precisely the linearized gauge symmetries of the Cartan-Weyl action. Equivalently, these are the gauge symmetries of its quadratic approximation:
\begin{equation}
\int_X \left( \bar e^a d_X \bar\omega^{bc} \Vol_{abc}+\bar\omega^a{}_c\bar\omega^{cb}\Vol_{ab}\right)\,,
\end{equation}
which is also known as the frame-like action of massless spin-2 field.
Thus we conclude that the master action $L^{BV}$ is a proper solution to the master equations and hence provides a correct BV formulation of Einstein gravity.

\section{The origin of the target space structures}

As we have seen the target space $(\algg[1],q,\omega^{\algg[1]})$ in a natural way defines  the complete BV formulation of general relativity. One may wonder how does this target space arise from the conventional formulation of gravity. This was mostly explained in~\cite{Grigoriev:2016wmk} but the relation between the presymplectic 2-form and the  BV antibracket was somewhat implicit. Here we give missing details.

Suppose we start with the BV-BRST complex of the metric gravity. The set of fields is given by the metric $g^{ab}$, diffeomorphism ghosts $\xi^a$ and their canonically conjugate antifields $g^*_{ab}$ and $\xi^*_{a}$. More geometrically, these variables are coordinates on the fiber $F$ of the underlying bundle $F\times X \to X$. The BV-BRST complex is given by local horizontal forms on the associated jet-bundle $J_X(F)$ equipped with the BRST differential and the horizontal differential $\dh$, for more details see \cite{Barnich:1995ap}. 

In the jet-bundle approach the standard symplectic structure of BV formulation  is given by
\begin{equation}
\omega^{sBV}=(dx)^n ( \dv g^{ab}\wedge \dv g^*_{ab}+ \dv \xi^{a}\wedge \dv \xi^*_{a})\,.
\end{equation}
It is of ghost degree $-1$ and horizontal form degree $n$. The BRST differential is an evolutionary vector field satisfying
\begin{equation}
\ii_s \omega^{sBV}=\dv L^{sBV}+\dh (\cdot)
\end{equation}
where $L^{sBV}$ is the integrand of the BV master action. It follows
\begin{equation}
    L_s \omega^{sBV}=\dh (\cdot)\,.
\end{equation}
Moreover, at least locally one can complete $\omega^{sBV}$ to a cocycle of the total BRST differential $\tilde s=\dh+s$:
\begin{equation}
    \omega^{tBV}=\omega^{sBV}+\omega^{sBV}_{n-1}+\ldots
\end{equation}
where $\omega^{sBV}_{k}$ has horizontal form degree $k$ and ghost degree $n-1-k$. 

Now we use the standard statement that for diffeomorphism-invariant systems, and gravity in particular, by changing variable one can bring the total BV-BRST complex to the form where $\tilde s=d_X+s$ and then eliminate 
$x^a,dx^a$ as contractible pairs, see e.g.~\cite{Barnich:1995ap,Brandt:1996mh,Barnich:2010sw} and references therein for more details. In more geometrical terms, this means that the underlying BV jet-bundle seen as a $Q$-bundle over $T[1]X$ is locally-trivial (see~\cite{Grigoriev:2019ojp} for more details).

Furthermore, eliminating further contractible pairs the total BRST complex reduces~\cite{Barnich:1995ap,Brandt:1996mh,Barnich:2010sw} to the minimal BRST complex of functions on the reduced  ghost-extended stationary surface $\cE$ which can be coordinatized by
\begin{equation}
 \xi^a,~~\rho^{ab}\,, \qquad W^{ab}_{cd},~~W^{ab}_{cd;c_1},~~\ldots ~~W^{ab}_{cd;c_1\ldots c_l}, ~~\ldots
\end{equation} 
where the first group of variables have ghost-degree $1$ and the second $0$.\footnote{Here we intentionally used the same notations as for coordinates on $\algg[1]$ to anticipate the relation between $\cE$ and $\algg[1]$.} Variables $\xi^a,\rho^{ab}$ originate from the diffeomorphism ghost and its antisymmetrized derivatives, while $W$-variables are related to the Weyl tensor and its algebraically-independent covariant derivatives restricted to the stationary surface.  Note that $W^{ab}_{cd;\ldots}$ variables can be chosen totally traceless. 

Upon the elimination of contractible pairs, the $\tilde s$-differential determines a differential  $q$ on $\cE$. Its complete explicit form in terms of intrinsic coordinates on $\cE$ is not known except in the sector of ghost degree $1$ variables:
\begin{equation}
\label{qgravred}
 q \xi^a=\xi^a{}_c\, \xi^c\,, \qquad q\rho^{ab}=\rho^{a}{}_c \,\rho^{cb}+\lambda \xi^a\xi^b+\xi^c\xi ^d W^{ab}_{cd}\,, \qquad \ldots\,,
\end{equation}
Furthermore,  $qW^{ab}_{cd;\ldots}$ is again proportional to $W^{ab}_{cd;\ldots}$, see e.g.~\cite{Brandt:1997iu,Barnich:2010sw}.

Local functions on $\cE$ equipped with $q$ form the minimal BRST complex for Einstein gravity.~\footnote{Note that the supermanifold $\tilde\cE$ of this variables equipped with $Q$-structure encodes all the information of the initial gauge theory. Indeed, as was shown in~\cite{Barnich:2010sw}, taking $\tilde\cE$ as a target space of the AKSZ sigma-model gives an equivalent formulation of the initial system at the level of equations of motion so that the system is fully reconstructed. It's equations of motion and gauge symmetries are precisely those of the minimal unfolded formulation~\cite{Vasiliev:1989xz,Vasiliev:2005zu}. Note however, that the explicit form of $q$ and hence of the unfolded equations of motion is not known in the intrinsic terms of $\cE$ but its existence, structure, and the implicit definition are easily arrived at  starting from the standard BV-BRST complex.   Analogous considerations apply to generic gauge theories though in contrast to gravity for linear theories formulations of this sort can be quite concise and explicit. See~\cite{Barnich:2010sw} and references therein for further details.} Because the minimal complex is an equivalent reduction of the initial one the initial $\tilde s$-cocycle $\omega^{tBV}$ gives rise to the respective $q$-cocycle  $\omega^{\cE}$ in the space of closed $2$-forms on $\cE$ of total ghost degree $n-1$. In fact the general structure of such form is rather restricted and it can be shown\footnote{To see this one can observe that elimination of the contractible pairs preserves the filtration by the order of derivatives. Together with the $q$-invariance  and the ghost degree conditions  this essentially fixes the form of the presymplectic structure. As an independent consistency check one can take the analogous presymplectic structure for linearized gravity computed in~\cite{Sharapov:2016sgx} and observe that the only component surviving the reduction is the linearization of $\omega^{\algg[1]}$.} that $\omega^{\cE}$ is precisely $\omega^{\algg[1]}$ trivially extended from $\algg[1]$ to $\cE$ (note that locally $\cE$ is a product of $\algg[1]$ and the space of Weyl tensors $W^{ab}_{cd\ldots}$). 

Now one can consider a presymplectic AKSZ sigma model with the target space $(\cE,q,\omega^{\cE})$ and try to reduce to the symplectic quotient. In the case at hand it is convenient to do it in two steps. In the first step one considers a distribution on $\cE$ generated by vector fields $\dl{W^{ab}_{cd;\ldots}}$ which obviously belong to the kernel distribution. The quotient can be realized as the surface $W^{ab}_{cd;\ldots}=0$ which is precisely $\algg[1]$ with the induced two form being $\omega^{\algg[1]}$
and the $Q$-structure being the restriction of $q$ to the surface.
In this way one systematically rederives the presymplectic  AKSZ formulation of Section~\bref{sec:psymp} starting from the conventional BV-BRST formulation of gravity.

It turns out that analogous considerations apply to a rather wide
class of gauge theories, giving a more precise understanding of the supergeometrical structures underlying their frame-like Lagrangians and BV formulations. Various examples of such presymplectic AKSZ formulations can be found in~\cite{Alkalaev:2013hta}.

\section{BFV phase space from presymplectic AKSZ}

Given an AKSZ model on a space-time manifold of the form $X=\Sigma\times \fR^1$, where $\Sigma$ corresponds to spatial slice and  $\fR^1$ to the time-line it is known~\cite{Grigoriev:1999ys,Barnich:2003wj,Grigoriev:2012xg} (see also~\cite{Cattaneo:2012qu,Ikeda:2019czt} for related developments, generalizations and applications) that its BFV formulation is given by an AKSZ sigma model restricted to $T[1]\Sigma$. The change of the dimension of the source space shifts by 1 the degree of the AKSZ action and the symplectic structure so that indeed such BFV-AKSZ sigma models defines a BFV formulation. 

That the constructed BFV formulation is correct immediately follows from the AKSZ formulation~\cite{Grigoriev:1999qz} of the standard construction~\cite{Fisch:1989rm,Dresse:1990ba} of the BV formulation from the BFV one.  More precisely, the 1d AKSZ sigma model with the target space being the above BFV-AKSZ sigma model can be identified with the initial AKSZ sigma model. Indeed, the space of supermaps from $T[1]\fR^1$ to $Smaps(T[1]\Sigma,\algg[1])$ is naturally identified with $Smaps(T[1](\Sigma\times \fR^1),\algg[1])$ and it is easy to check that the respective AKSZ structures coincide.

Let us consider the BFV version of the above presymplectic AKSZ sigma-model, which is obtained by replacing $T[1]X$ with $T[1]\Sigma$. The construction of an analog $M_H=Smaps(T_x[1]\Sigma,\algg[1])$ of the graded manifold $M$ is straightforward.  In so doing the coordinates  $e^a_i(x)$ entering $\xi^a(\theta)$ as $e^a_i(x)\theta^i$ are assumed to be such that $e^k_i(x)$, $k=1,\ldots,n-1$ is invertible (here and below we denote by $\theta^i$ the coordinates on the fibers of $T_x[1]\Sigma$). Note that now the presymplectic structure $\omega^{M_H}$ and ``Hamiltonian'' $\cH_H$ have ghost degree $0$ and $1$ respectively, i.e. are shifted by $1$ as compared to BV-AKSZ sigma model.

Just like $M$, $M_H$ is a graded presymplectic manifold. To see that it gives rise to a symplectic one let us apply Lemma~\bref{lemma:w-sub} taking as $M_H^0$ a submanifold determined by $\xi^a=0,\st{2}{\xi}{}^a_{ij}=0\,, \ldots\,,,\st{n-1}{\xi}{}^a_{i_1\ldots i_{n-1}}=0$. Then fix a generic point of $M_H^0$ and adjust the basis in $T_x[1]\Sigma$ and $\algg[1]$ such that $e^0_i=0$ and $e^a_i=\delta^a_i$. At this point the presymplectic structure can be written as:
\begin{equation}
\label{omega-p-h}
\omega^{M_H}_p = 
de^i_j\wedge d\omega^{0j}_i+
de^0_j\wedge d\omega^{j}+
d\xi^j \wedge d \st{2}{\rho}{}_j+
d\xi^0 \wedge d \st{2}{\rho}+
d\rho^{k0}\wedge d\st{2}{\xi}{}_{k}+
d\rho^{kj}\wedge d\st{2}{\xi}{}^0_{kj}
\end{equation}
where $\omega^{j},\st{2}{\rho}{}_{j}$, $\st{2}{\rho}$, $\st{2}{\xi}{}_{k}$ parameterize the following components:
\begin{equation}
\omega^{kj}_k\,, \qquad 
\st{2}{\rho}{}^{0k}_{jk}\,, \qquad 
\st{2}{\rho}{}^{jk}_{jk}\,, \qquad 
\st{2}{\xi}{}^{k}_{jk}
\end{equation}
and we took a liberty to redefine some of the components by constant factors. The coordinates along which $\omega^{M_H}$  degenerates are:
\begin{equation}
\st{3}{\xi}{}^{a|}\,, \qquad 
\st{3}{\rho}{}^{ab|}\,, \qquad 
\st{2}{\xi}{}^{k|i}e^j_i+(ij)\,, \qquad 
\st{2}{\rho}{}^{ak|i}e^j_i+(ij)\,, \qquad 
\omega^{ik|mn}e_m^je_n^l+(ij)(kl)\,.
\end{equation}
Lemma~\bref{lemma:w-sub} then implies that a submanifold $N_H$ where these coordinates vanish is symplectic. Moreover, the spectrum of the coordinates along which \eqref{omega-p-h} is nondegenerate precisely corresponds to the coordinates of the BFV phase space for Cartan-Weyl formulation of gravity.  More precisely, the BFV phase space we have arrived at is the one discussed recently in~\cite{Canepa:2020ujx}. Note that there exist alternative (but equivalent) BFV formulations which are related through elimination of BFV analogs~\cite{Grigoriev:2010ic,Grigoriev:2012xg} of the conventional generalized auxiliary fields~\cite{Dresse:1990dj}.

Let us consider the body $M_H^{body} \subset M_H$, obtained by setting to zero all the coordinates of nonvanishing degree. In contrast to the BV case the restriction of symplectic structure to the body is nonvanishing. More precisely, it gives rise to the phase space symplectic structure of the underlying constrained system. Indeed, setting all the nonvanishing degree coordinates to zero the resulting presymplectic form reads as:\footnote{Another way, employed recently in~\cite{Canepa:2020ujx} to arrive at this presymplectic structure  is to start with Cartan-Weyl action and find the presymplectic current (see e.g.~\cite{Crnkovic:1986ex,Khavkine2012,Sharapov:2016qne,Grigoriev:2016wmk}) that it defines on the stationary surface. The present derivation of this presymplectic structure from that on $\algg[1]$ was somewhat implicitly already in~\cite{Alkalaev:2013hta}}
\begin{equation}
\omega^{M_H^{body}}=e^a_i d e^b_j \wedge d\omega^{cd}_k \epsilon^{ijk}\epsilon_{abcd}\,.
\end{equation}
It is clear that this form is degenerate and moreover is a regular presymplectic one.  Employing the basis where $e^0_i=0$ and $e^i_j=\delta^i_j$ it is clear that the symplectic form on the quotient is precisely the restriction of~\eqref{omega-p-h} to the body of $N_H$.  The symplectic 2-form induced on the quotient gives the phase-space symplectic structure.

Note that the quotient can be described more invariantly~\cite{Canepa:2020ujx}. Namely, consider the following equivalence relation on the space with coordinates $\omega^{ab}_i$:
\begin{equation}
    \omega^{ab}_i\sim \omega^{ab}_i+v^{ab}_i\,, \qquad 
    v^{ab}_i e^c_j\epsilon^{ijk}\epsilon_{abcd}=0
\end{equation}
Using the adapted basis it is easy to check that this equivalence relation precisely removes the trace-free component of $\omega^{jl}_i$ leading to the restriction of~\eqref{omega-p-h} to the body of $N_H$.

Having chosen symplectic $N_H \subset M_H$ let us consider the BFV phase space $J_\Sigma(N_H)$. It is equipped with the induced symplectic structure and the local functional $S^{BFV}$ obtained by restricting AKSZ BFV charge from $J_\Sigma(M_H)$ to $J_\Sigma(N_H)$. The restrictions of ghost degree 1 coordinates $\xi^a$ and $\rho^{ab}$ to $N_H$ remain independent and are to be interpreted as ghost variables. The terms in $S^{BFV}$ linear in $\xi^a$ and $\rho^{ab}$ read respectively as:
\begin{equation}
    \int_\Sigma \xi^a \epsilon_{abcd}e^b (d_\Sigma\omega^{cd}+\omega^{c}{}_e \omega^{ed})
\,, \qquad 
    \int_\Sigma \rho^{ab} \epsilon_{abcd}e^c (d_\Sigma e^d+\omega^d{}_f e^f) \,.
    %
\end{equation}
The coefficients are precisely the 1-st class constraints
encoded in the Cartan-Weyl action.  It follows $S^{BFV}$ have the structure similar to that of the proper BRST charge of the theory in question.  What does not follow from the above considerations is that $S^{BFV}$ satisfies master equation on $J_\Sigma(N_H)$. If $M_H$ were a regular presymplectic manifold this would follow just like in BV case. However, the presymplectic structure on $M_H$ is in fact not regular and the detailed analysis of the presymplectic BFV-AKSZ formulation of gravity will be performed elsewhere.

Let us only comment on the relation between BV and BFV formulations arising from the presymplectic AKSZ.  It is easy to see that $M$ does not coincide with $T^*[-1]M_H$ and hence these BV and BFV formulations are not related through the usual construction~\cite{Fisch:1989rm,Dresse:1990ba,Grigoriev:1999qz}. More precisely, the BV formulation obtained from the above BFV in this way is a certain equivalent reduction of the standard one. Indeed, because not all components of the Lorentz connection are independent coordinates on $M_H$ it is easy to see that the same applies to $T^*M_H$ whose coordinates are the BV fields. At the same time in the standard BV formulation for the Cartan-Weyl action all the components of the Lorentz connection are  independent fields. This subtlety seems to be directly related to the discrepancy observed and investigated recently in~\cite{Canepa:2020rhu,Canepa:2020ujx}. However, the respective presymplectic BV and BFV AKSZ sigma models described in this work are obviously related via a straightforward presymplectic extension of the 1d AKSZ construction of~\cite{Grigoriev:1999qz}.

\section{Conclusions}

By concentrating on the example of  general relativity  we have demonstrated that presymplectic AKSZ-type sigma models naturally encode BV as well as BFV formulations in a rather concise and geometrical way. In so doing we have uncovered an interesting supergeometrical structures underlying the BV formulation of the frame-like gravity. This makes more explicit the deep relation between the underlying Cartan geometry and the BV formulation of gravity.  

The present construction can be regarded as the BV extension of the so-called intrinsic Lagrangians~\cite{Grigoriev:2016wmk}, which are natural 1st order Lagrangian defined in terms of the equation manifold (stationary surface of the theory) equipped with the horizontal differential and the presymplectic current.

It is important to stress that for various applications the presymplectic AKSZ formulation can be used in place of conventional the conventional AKSZ. For instance, the formal path integral can be written just in terms of the presymplectic AKSZ action. The only difference is that some additional gauge-fixing conditions taking care of the kernel of the presymplectic 2-form are to be implemented in the path integral. Analogous remark applies to the presymplectic generalization of the generic (not necessarily AKSZ) BV fomalism.

Possible further developments include the extension of the present considerations to general local gauge theories including those which are not diffeomorphism-invariant. This can be naturally done using the language of gauge PDEs~\cite{Grigoriev:2019ojp} equipped with the compatible presymplectic structures. There also remains to investigate further the BFV interpretation of the presymplectic AKSZ formulation of gravity as well as the presymplectic AKSZ version of the relation between its BV and BFV descriptions.

An attractive feature of the AKSZ formalism is that it makes  manifest the relation between the bulk theory and the theory induced on the boundary. This feature already manifest itself in that Hamiltonian formulation (seen as a boundary theory induced on the surface of the initial data) is obtained by simply pulling back the AKSZ model to the boundary~\cite{Grigoriev:1999ys,Barnich:2003wj,Grigoriev:2010ic,Cattaneo:2012qu,Grigoriev:2012xg}. Applications to more general situations including holographic relations can be found in~\cite{Bekaert:2012vt,Bekaert:2013zya,Grigoriev:2018wrx,Mnev:2019ejh,Rejzner:2020xid,Canepa:2020rhu}. Let us also mention recent works~\cite{Freidel:2020ayo,Freidel:2020svx}, where the presymplectic structure (also known as the presymplectic current, see e.g.~\cite{Crnkovic:1986ex,Khavkine2012,Sharapov:2016qne}) induced on the space of solutions to frame-like gravity is employed in the study of its boundary structure.

\section*{Acknowledgments}
\label{sec:Aknowledgements}

M.G. acknowledges discussions with G.~Barnich and V.~Gritsaenko. He is also grateful to A.~Cattaneo, G.~Caneppa, M.~Schiavina for the discussions and for attracting attention to their recent  related works~\cite{Canepa:2020rhu,Canepa:2020ujx}. The work of M.G. was supported in part by the Russian Science Foundation grant 18-72-10123. The research of A.K. was supported by the grant no. 18-00496S of the Czech Science Foundation.
\appendix

\section{The structure of the kernel}
\label{app:VF-kernel}

Here we show that prolongations of the vector fields~\eqref{VF} determine the kernel of $\omega^M$ at generic $p\in M_0$. To begin with it is easy to see that $\hat X^4_{a}$ and  $\hat Y^4_{ab}$ exhaust the kernel of $\omega^M$ in the sector of $\st{4}{\xi}{}^a_{\mu\nu\rho\sigma}$
and $\st{4}{\xi}{}^{ab}_{\mu\nu\rho\sigma}$.

Consider then $\hat X^3_{ab}$. At $p\in M_0$ one has:
\begin{equation}
\hat X^3_{ab}=\epsilon_{acdf}e^c_\mu e^d_\nu e^f_\rho \dl{\st{3}{\xi}{}^b_{\mu\nu\rho}}+(ab)
\end{equation}
where $(ab)$ denote terms symmetrizing the expression in $a$ and $b$. Using the basis where $e^a_\mu\delta^a_\mu$ and introducing notation $\st{3}{\xi}{}^{a|e}$ for $\epsilon^{efcd}\st{3}{\xi}{}^a_{fcd}$ one finds:
\begin{equation}
\hat X^3_{ab}=\dl{\st{3}{\xi}{}^{a|b}}+(ab)\,.
\end{equation}
Using symmetric and  antisymmetric components $\hat \xi_S^{ab}$ and $\hat \xi_A^{ab}$  as new coordinates it is easy to see that $\hat X^3_{ab}=\dl{{\hat\xi}{}^{ab}_S}$ while $\omega^M$ is nondegenerate on $\dl{{\hat\xi}{}^{ab}_A}$. These later coordinates parameterize the antifields conjugated to $\rho^{ab}$, cf. \eqref{omega-p}.

Now we turn to $\hat Y^{2}_{abcd}$. As before restricting to a generic point of $M_0$ and using a basis where $e^a_\mu=\delta^a_\mu$ introduce new coordinates replacing $\rho$:
\begin{equation}
    \bar\rho^{abcd}=\epsilon^{ab\mu\nu}\st{2}{\rho}{}^{cd}_{\mu\nu}
\end{equation}
This is an invertible change of coordinates. In the new coordinates:
\begin{equation}
\hat Y^{2}_{abcd}=\dl{\bar\rho^{abcd}}+(ac)(bd)
\end{equation}
where $(ac)(bd)$ denote 3 terms symmetrizing $ac$ and $bd$. 

In terms of $\bar\rho$ the trace $\st{2}{\rho}{}^{ab}_{cb}$ can be expressed in terms of $\epsilon_{abcd}\bar\rho^{ebcd}$ and one gets
\begin{equation}
\hat Y^{2}_{abcd} \left( \epsilon_{\mu\nu\rho\sigma}\bar\rho^{\alpha \nu\rho\sigma}\right)
\end{equation}
so that the trace is not in the kernel of the presymplectic structure. The complementary component is described by $\rho$ satisfying $\epsilon_{abcd}\bar\rho^{ebcd}=0$, i.e. having the symmetry type associated to rectangular YT. But this is precisely the tensor structure of $\hat Y^2_{abcd}$. Using new coordinate system $\hat\rho^{abcd},\hat\rho^{ab}$, where $\hat\rho^{abcd}$ is the component of $\bar\rho^{abcd}$ that have symmetry structure described by the rectangular YT while $\hat\rho^{ab}$ parameterize the trace, one finds:
\begin{equation}
\hat Y^2_{abcd}=\dl{\hat\rho^{abcd}}+(ac)(bd)    
\end{equation}
so that indeed $\dl{\hat\rho^{abcd}}$ are in the kernel of the symplectic form. It is easy to check that among $\dl{\hat\rho^{ab}}$ there are no zero vectors of the presymplectic form. In fact $\hat\rho^{ab}$ parameterize antifields associated to $e^a_\mu$.

The remaining fields are $\hat Y^3_{abc}$. Introduce new parameterization of $\st{3}{\rho}{}_{\mu\nu\rho}^{ab}$ in terms of $\bar\rho^{c|ab}$ proportional to $\epsilon^{c\mu\nu\rho} \st{3}{\rho}{}_{\mu\nu\rho}^{ab}$ so that
\begin{equation}
\hat Y^3_{abc}=\dl{\bar\rho^{a|bc}}+(ab)
\end{equation}
$\hat Y^3_{abc}$ contains two irreducible components: $\epsilon_{abcd} \bar\rho^{a|bc}$ and ${\bar\rho^{a|bc}}+(ab)$. The first is precisely the double trace $\st{3}{\rho}{}_{abd}^{ab}$ which satisfies $Y^3_{abc}\st{3}{\rho}{}_{efd}^{ef}=0$. The second ones give rise to $Y^3_{abc}$. Presymplectic form $\omega^M$ is nondegenerate along $\st{3}{\rho}{}_{abd}^{ab}$.

\footnotesize
\bibliography{HSmaster}
\end{document}